\documentclass [12pt] {article}

\usepackage[utf8]{inputenc}

\usepackage{amsmath}  
\usepackage{amsfonts}
\usepackage{amssymb}
\usepackage{amsthm}
\usepackage{cite}
\usepackage{fullpage}
\usepackage[]{nicefrac}
\usepackage{bbm}
\usepackage{qcircuit}
\usepackage{authblk}
\usepackage{color}
\usepackage{float}
\usepackage[normalem]{ulem}
\usepackage{cancel}
\usepackage{graphicx}
\usepackage{varwidth}
\usepackage{tikz,tkz-graph}
\usepackage{enumitem}

\newcommand{\ignore}[1]{}
\newcommand{\ket}[1]{\left|#1\right\rangle}
\newcommand{\bra}[1]{\left\langle#1\right|}
\newcommand{\braket}[2]{\left\langle #1| #2 \right\rangle}
\newcommand{\ketbra}[2]{ \left| #1 \right\rangle\left\langle #2 \right|}
\newcommand{\prnt}[1]{\left( #1 \right)}
\newcommand{\prntt}[1]{\left[ #1 \right]}\newcommand{\prnttt}[1]{\left\{ #1 \right\}}
\newcommand{\abs}[1]{\left| #1 \right|}

\newcommand{\yanote}[1]{}%[1]{\textcolor{red}{{\bf (Yosi:}{#1}{\bf )}}}

\def\PSPACE{\ensuremath{\textsc{PSPACE}}}
\def\PP{\ensuremath{\textsc{PP}}}
\def\poly{\ensuremath{\textsc{poly}}}

\newtheorem{theorem}{Theorem}
\newtheorem{fact}{Fact}
\newtheorem{corollary}{Corollary}%[theorem]
\theoremstyle{definition}
\newtheorem{definition}{Definition}
\newtheorem{lemma}{Lemma}

\title{On the Hardness of Detecting Macroscopic Superpositions}
\author[1]{Scott Aaronson}
\author[1]{Yosi Atia}
\author[2,3]{Leonard Susskind}
\affil[1]{\small{Department of
Computer Science, University of Texas at Austin, Austin, TX, USA}}
\affil[2]{\small{SITP, Stanford University, Stanford, CA 94305, USA}}
\affil[3]{\small{Google, Mountain View, CA 94043, USA}}
\date{}

\begin{document}

\bibliographystyle{unsrt}
\maketitle
\abstract{
When is decoherence ``effectively irreversible''? \ Here we examine this
central question of quantum foundations using the tools of quantum
computational complexity. \ We prove that, if one had a quantum circuit
to determine if a system was in an equal superposition of two
orthogonal states (for example, the $\ket{\mathrm{Alive}}$ and $\ket{\mathrm{Dead}}$ states of
Schr\"{o}dinger's cat), then with only a slightly larger circuit, one
could also \textit{swap} the two states (e.g., bring a dead cat back
to life). \ In other words, observing interference between the $\ket{\mathrm{Alive}}$
and $\ket{\mathrm{Dead}}$ states is a ``necromancy-hard'' problem, technologically
infeasible in any world where death is permanent. \ As for the
converse statement (i.e., ability to swap implies ability to detect
interference), we show that it holds modulo a single exception,
involving unitaries that (for example) map $\ket{\mathrm{Alive}}$ to $\ket{\mathrm{Dead}}$ but
$\ket{\mathrm{Dead}}$ to $-\ket{\mathrm{Alive}}$. \ We also show that
these statements are robust---i.e., even a \textit{partial} ability to
observe interference implies partial swapping ability, and vice versa. \ Finally, without relying on any unproved
complexity conjectures, we show that all of these results are
quantitatively tight. \ Our results have possible implications for the
state dependence of observables in quantum gravity, the subject that
originally motivated this study.
}

\section{Introduction}
Schr\"odinger's cat famously raised the question: how large does a quantum state have to be, before we can take it to represent actual events rather than just potentialities? \ In practice, the larger a state, the harder it is to keep track of all of its degrees of freedom, and the harder it is to prevent it from interacting with its environment; both effects can quickly make it infeasible to observe quantum coherence between different branches of the state. \ But is there some principled criterion for saying when two branches have become ``macroscopically distinct,'' in the sense that observing interference between them is now so technologically intractable that one might as well speak in terms of a ``collapse'' having happened?\footnote{Similarly, one of the questions raised in the context of the AdS/CFT correspondence is the ``state dependence of observables'' \cite{PR14, Harlow14}. \ For example, in the discussion of black hole firewalls, which bulk operator an observer can apply depends on the spacetime background---e.g., is there a black hole or no black hole? \ This means that the set of measurements one can perform would depend on the quantum state of the spacetime background. \ While this seems to make little sense in the context of standard quantum mechanics, it is really about the dictionary between the standard quantum mechanics of the boundary holographic description and the incompletely understood bulk description of phenomena behind the horizon. \ Our results show that, if an observer cannot efficiently map one spacetime branch to the other, then she also cannot efficiently measure a superposition of the two branches, and thus effectively sees one or the other. }

Brown and Susskind \cite{BS18} conjectured that \emph{relative complexity}, as defined below, characterizes how hard it is to observe coherence between two orthogonal quantum states $\ket x$ and $\ket y$---in the sense of performing a measurement that accepts the superposition $\frac{\ket x + \ket y}{\sqrt{2}}$ and rejects the superposition $\frac{\ket x - \ket y}{\sqrt{2}}$ with high probability. \ (We note that, by a convexity argument, this is essentially equivalent to distinguishing either of those two superpositions from the \textit{classical mixture} $\frac{1}{2} (\ket x \bra x + \ket y \bra y )$.)
\begin{definition} \label{def:RelCmplx} [Relative complexity, adapted from  \cite{Susskind2018}] ~\\Let $\ket{x},\ket{y}$ be two $n$-qubit pure quantum states. \ Their relative complexity, $\mathcal C_\varepsilon (\ket x,\ket y)$, is the minimal number of gates in a circuit $C$ such that $$\big|\bra{y}\bra{0\dots0}C\ket{x}\ket{0\dots0}\big|^2\ge 1- \varepsilon.$$
The gates are chosen from an arbitrary fixed universal set of 1-qubit and 2-qubit gates;
ancilla qubits are allowed as long as they return the $\ket{0\dots0}$ state. \end{definition}

We will omit the dependence on $\varepsilon$ when it is not necessary. \ It is easy to see that $\mathcal C_\varepsilon$ is a metric: it is symmetric; it is zero iff $\ket{x}=\ket y$; and it satisfies the triangle inequality (with $\varepsilon$ increased):
$$
    \mathcal C(\ket x , \ket y) + \mathcal C(\ket y , \ket z) \ge \mathcal C(\ket x , \ket z) 
$$
By a counting argument, the relative complexity of almost any pair of  $n$-qubit states is $2^{\Omega(n)}$. \ Importantly, two states could be orthogonal, but still extremely close in relative complexity distance: for example, $\ket{0^n}$ and $\ket{0^{n-1}1}$. \ Conversely, two states could be close in $\ell_2$-distance but far in relative complexity: for example, consider the states $\ket{0^n}$ and $\ket{\phi}=\sqrt{1-\varepsilon}\ket{0^n}+\sqrt{\varepsilon}\ket{\psi}$, where $\ket\psi$ is Haar-random and $\varepsilon$ is small. The $\ell_2$-distance is between them is $O(\varepsilon)$, but for instance, $\mathcal C_{\varepsilon/4}(\ket \psi, \ket \phi)$ is exponential.

In this paper we also consider a slightly different notion, the \emph{swap complexity} of two states, which is at least as large as their relative complexity but could be larger.
\begin{definition} \label{def:swapability}
Let $\ket{x},\ket y$ be two quantum states. \ Their swap complexity, $\mathcal{S}_\varepsilon (\ket{x},\ket{y})$, is the minimal number of gates in a circuit $C$ such that
$$
    \frac{\abs{\bra x \bra{0\dots0} C \ket y \ket{0\dots0}+ \bra y \bra{0\dots0} C \ket x \ket{0\dots0}}}{2} \ge 1-\varepsilon
$$
Ancilla qubits are allowed as long as they return the $\ket{0\dots0}$ state.
\end{definition}

We show that the complexity of transforming $\ket x \leftrightarrow \ket y$, is (up to a constant factor) equal to the complexity of perfectly distinguishing between $\ket \psi=\frac{\ket x+\ket{y}}{2}$ and $\ket \phi= \frac{\ket x - \ket y}{2}$. \ The swap complexity is at least the relative complexity, and as we show later, there are cases where the swap complexity is exponential while the relative complexity is $O(1)$.

Going further, we show an equivalence even between \textit{approximate} versions of swapping and distinguishing. \ Without relying on unproved complexity conjectures, we also show that our approximate equivalence theorem is optimal, in the sense that a result with better error parameters would be false.

Of course, qualitatively similar observations had been made before, but as far as we know, never in the sharp form here, which seems to require the formal notion of quantum circuit complexity or something similar. \ As one example, Aharonov and Rohrlich \cite[Chapter 9]{AR08} pointed out that the ability to measure a cat in the $\{\ket{\mathrm{Alive}}\pm \ket{\mathrm{Dead}} \}$ basis would imply the ability to revive a dead cat with success probability $1/2$, a weaker statement than what we show here.

The equivalence between swap complexity and  observing coherence is interesting in the context of the foundations of quantum mechanics. \ For example, in the Schr\"odinger's cat experiment, having the technological ability to \emph{detect} that the cat was in superposition state at all, implies having the ability 
to perform a unitary that revives a dead cat, an ability that one could call ``quantum necromancy.''   \ In other words, our results show that, if reviving a dead cat is considered ``hard''---for essentially any reasonable definition of ``hard''---then distinguishing Schr\"{o}dinger's cat from a classical mixture is ``hard'' in that same sense, and the cat can be treated as effectively decohered. \ Similarly, in the Wigner's Friend thought experiment \cite{Wigner55,Bass71}, if Wigner can detect that his friend is in superposition then he can also swap his friend's mental states.

\section{Main Result}
\subsection{Perfect case}
We start by proving the equivalence (in circuit complexity) between a perfect swapper of two orthogonal states, and a perfect distinguisher for the corresponding conjugate states.
\begin{theorem} \label{thm:perfect}
Let $\ket x,\ket y$ be $n$-qubit orthogonal quantum states, and let $\ket{\psi}=\frac{\ket x + \ket y}{\sqrt 2}$ and $\ket{\phi}=\frac{\ket x - \ket y}{\sqrt 2}$. \ The following two statements are equivalent:
\begin{enumerate} [label=(\roman*)]
    \item There is a unitary $U$ such that $U\ket x = \ket y$ and $U\ket y = \ket x$ with circuit complexity $O(T(n))$.
    \item There is a unitary which perfectly distinguishes between $\ket{\psi}$ and $\ket \phi$ with circuit complexity $O(T(n))$. 
\end{enumerate}
Indeed, one can perfectly distinguish $\ket{\psi}$ from $\ket\phi$ given a single black-box access to a controlled swap of $\ket x$ and $\ket y$, and one can swap $\ket x$ and $\ket y$ given a single black-box access to a unitary $A$ that simulates a measurement perfectly distinguishing $\ket{\psi}$ from $\ket{\phi}$, as well as a single black-box access to $A^\dagger$.
\end{theorem}

\begin{proof} ~\\
$(i)\rightarrow (ii)$\\
%\yanote{there's some subtlety here, since we use ctrl-U and not U}
For distinguishing between $\ket \psi$ and $\ket \phi$, apply $U$ to $\ket \psi$ or $\ket \phi$, conditioned on a $\ket +$ control qubit being $\ket 1$ (see Fig.\ \ref{fig:perfect1}). \ Then check whether the control qubit becomes $\ket{-}$. Only $O(1)$ gates are added to $U$, hence the complexity is still $O(T(n))$.\footnote{Note that in the definition of swap complexity (Definition \ref{def:swapability}), it was crucial that the ancilla qubits return to the all-$0$ state. \ Without that requirement, one can check that our construction of the distinguishing circuit in Figure \ref{fig:perfect1} would fail, because the control qubit would be entangled with the ancilla qubits. \ By contrast, the circuit for the other direction of the proof (see Figure \ref{fig:perfect2}) is insensitive to the ancilla state. \ (We thank Daniel Gottesman for this observation.)}
\begin{figure} [H] 
{
\[
\Qcircuit @C=0.7em @R=1em {
\lstick{\ket{+}}  & \ctrl{1}{phase} & \qw & \gate{H} & \qw &\meter
\\
\lstick{\ket{\psi} ~\mathrm{or}~ \ket{\phi}} & \gate{U} & \qw&\qw&\qw&\qw 
}  
\]
}
\caption{\label{fig:perfect1}A circuit distinguishing $\ket \psi$ from $\ket \phi$ using a unitary $U$ that swaps $\ket{x}$ with $\ket{y}$.}
\end{figure}

\noindent $(ii)\rightarrow (i)$\\
Suppose we had a unitary $A$ such that $A\ket{\psi}=\ket{0}\ket{g_\psi}$ and $A\ket{\phi}=\ket{1}\ket{g_\phi}$ where $\ket{g_\psi},\ket{g_\phi}$ are arbitrary states of the remaining $n-1$ qubits. \ Additionally, the circuit complexity of $A$ is $O(T(n))$. \ Then to swap $\ket{x}=\frac{\ket \psi+\ket \phi}{\sqrt{2}}$ and $\ket y=\frac{\ket \psi - \ket \phi}{\sqrt 2}$,  we just apply $A$, then apply a $Z$ gate on the first qubit, and finally uncompute by applying $A^\dagger$ (see Figure \ref{fig:perfect2}). \ Formally,
$$\ket x = \frac{\ket \psi+\ket \phi}{\sqrt{2}} \overset{A}{\rightarrow} \frac{\ket {0}\ket{g_\psi}+\ket {1}\ket{g_\phi}}{\sqrt{2}} \overset{Z_1}{\rightarrow} \frac{\ket {0}\ket{g_\psi}-\ket {1}\ket{g_\phi}}{\sqrt{2}} \overset{A^\dagger}{\rightarrow} \frac{\ket \psi-\ket \phi}{\sqrt{2}} = \ket y.$$
The total circuit complexity is twice the circuit complexity of $A$ and another $Z$ gate, which sums up to $O(T(n))$.

\begin{figure} [H] 
{
\[
\Qcircuit @C=0.7em @R=1.3em {
\lstick{}  &\qw&\multigate{3}{A} &\qw & \gate{Z}  & \qw & \multigate{3}{A^\dagger} & \qw 
\\
\lstick{}  &\qw &\ghost{A} &\qw &\qw & \qw& \ghost{A^\dagger}& \qw 
\\
{\begin{split}  &\ket{x}~\mathrm{or}~\ket{y}  \tiny{\begin{cases} 
\\
\\
\\
\\
\\
\\
\\
\\
\end{cases}} \\ &~\\ &~ \end{split}} ~~~~~~~~~~~~~~~~~~~~~~~~~ & ^{\vdots}~~~ &\ghost{A} &\qw & ^{\vdots} &  & \ghost{A^\dagger} & \qw 
\\
\lstick{}  &\qw &\ghost{A} &\qw &\qw & \qw& \ghost{A^\dagger}& \qw 
}  
\]
}
\caption{\label{fig:perfect2}A circuit implementing $\ket x \leftrightarrow \ket y$ using a distinguishing circuit $A$.}
\end{figure}
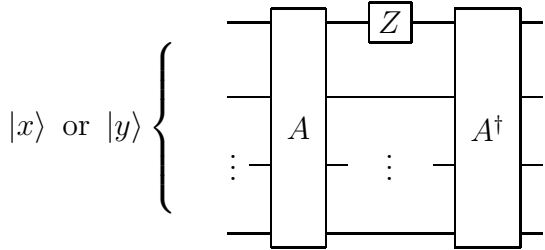
\end{proof}

Theorem \ref{thm:perfect} is related to a known equivalence between Hamiltonian simulation and energy measurement \cite{AA17}. \ In terms of \cite{AA17}, the swap unitary is a $\bar \sigma_x$ gate defined on the basis states $\ket{\bar 0}=\ket{x},\ket{\bar 1}=\ket{y}$. \ Similarly, distinguishing between $\ket \phi$ and $\ket \psi$ is equivalent to an energy measurement by the Hamiltonian $\bar \sigma_x$ with precision $\Delta E = 1$. \ Indeed, the circuit in Figure \ref{fig:perfect1} resembles the phase estimation (or energy measurement) circuit by Kitaev, Shen, and Vyalyi \cite{KSV02}, with $U$ as the simulation of the Hamiltonian $\bar \sigma_x$. \ Conversely, the circuit in Figure \ref{fig:perfect2} is a simulation of the Hamiltonian $\bar \sigma_x$ using $A$, where the latter separates $\frac{\bar {\ket{0}}+\bar{\ket{1}}}{\sqrt{2}}$ from $\frac{\bar {\ket{0}}-\bar{\ket{1}}}{\sqrt{2}}$. 

\subsection{Imperfect case}
Note that there are cases where $U$ efficiently maps $\ket x$ to $\ket y$, and a corresponding $U^\dagger$ maps $\ket y$ to $\ket x$, but the same $U$ doesn't do both. \ For example, $\ket x =\ket{0^n}$ and $\ket y=C \ket{0^n}$, where $C$ is some random quantum circuit (see Section \ref{sec:RelCvsSwapC} for a comparison of swap complexity and relative state complexity). \ In such cases, a natural question arises: how well can the circuit $C$ be used to distinguish $\ket \psi$ from $\ket \phi$?

More generally, we might wonder: if we have some ``imperfect'' or ``partial'' ability to swap $\ket x$ and $\ket y$ (for example, a unitary $U$ such that $U\ket x=\ket y$ but $U\ket y\neq \ket x$), what does that imply about our ability to measure the relative phase in $\frac{\ket x \pm \ket y}{\sqrt{2}}$? \ Conversely, what does an imperfect ability to measure the relative phase imply about our ability to swap?

Our main result is as follows:

\begin{theorem}\label{thm:main}~
\begin{enumerate}[label=(\roman*)]
    \item Let $\ket x,\ket y$ be orthogonal $n$-qubit states, and suppose that $\bra y  U \ket x=a$ and $\bra x  U \ket y=b$. \ Then using a single black-box access to controlled-$U$, plus $O(1)$ additional gates, we can distinguish $\ket \psi=\frac{\ket x + \ket y}{\sqrt 2}$ from $\ket \phi=\frac{\ket x - \ket y}{\sqrt 2}$ with bias $\Delta=\frac{\abs{a+b}}{2}$.
    \item Let $\ket \psi,\ket \phi$ be orthogonal $n$-qubit states, and suppose the procedure $A$ accepts $\ket \psi$ with probability $p$ and $\ket{\varphi}$ with probability $p-\Delta$, i.e.\ $A$ distinguishes $\ket \psi$ from $\ket \phi$ with bias $\Delta$. \ Then using a single black-box access each to $A$ and $A^{\dagger}$, and a single additional gate, we can apply a unitary $U$ such that
    $$
    \frac{\abs{\bra y  U \ket x + \bra x  U \ket y}}{2}=\Delta
    $$
    where $\ket x=\frac{\ket \psi + \ket \phi}{\sqrt 2}$ and  $\ket y=\frac{\ket \psi- \ket \phi}{\sqrt 2}$.
\end{enumerate}
\end{theorem}
Note that the parameters in the two parts of the theorem are equivalent. \ Given $a,b$ in Theorem \ref{thm:main}$(i)$, we can distinguish $\ket{\psi},\ket{\phi}$ with bias $\Delta =\frac{\abs{a+b}}{2}$,  while with the same distinguishibility bias $\Delta$, we can create a swap with parameters $\widetilde a,\widetilde b$ such that $\frac{\abs{\widetilde a+ \widetilde b}}{2}=\Delta$.

\begin{proof}
For part $(i)$, let $U$ be as follows:
\begin{equation*}
\begin{split}
    U\ket{x}&=\prnt{a\ket{y}+c\ket{x}+f\ket{w}}
    \\
    U\ket{y}&=\prnt{b\ket{x}+d\ket{y}+g\ket{z}},
\end{split}    
\end{equation*}
where $\ket{w},\ket{z}$ are states orthogonal to both $\ket{x}$ and  $\ket{y}$. \ We add a global phase to $U$, and denote it $\widetilde U = e^{i\theta}U$ (alternatively we initialize the ancilla qubit to  $\ket{0}+e^{i\theta}\ket{1}$).

Using the same procedure as in Fig. \ref{fig:perfect1} on the input $\ket{\psi}$, with $\widetilde U$, we get
\begin{equation*}
\begin{split}
    \Pr(\ket{+})&=\frac{1}{2} + \frac{1}{2}\mathrm{Re}\prnt{\bra \psi \widetilde U \ket \psi}
    \\
    &=\frac{1}{2} +\frac{1}{4} \mathrm{Re}\prntt{e^{i\theta}\prnt{\bra x +\bra y}\prnt{a\ket{y}+c\ket{x}+f\ket{w}+b\ket{x}+d\ket{y}+g\ket{z}}}
    \\
    &=\frac{1}{2}+\mathrm{Re}\prnt{e^{i\theta}\cdot \frac{a+b+c+d}{4}},
\end{split}
\end{equation*}
whereas on input $\ket \phi$, we get 
\begin{equation*}
\begin{split}
    \Pr(\ket{+})&=\frac{1}{2} + \frac{1}{2}\mathrm{Re}\prnt{\bra \phi \widetilde U \ket \phi}
    \\
    &=\frac{1}{2} +\frac{1}{4} \mathrm{Re}\prntt{e^{i\theta}\prnt{\bra x -\bra y}\prnt{a\ket{y}+c\ket{x}+f\ket{w}-b\ket{x}-d\ket{y}-g\ket{z}}}
    \\
    &=\frac{1}{2}+\mathrm{Re}\prnt{e^{i\theta}\cdot\frac{-a-b+c+d}{4}}.
\end{split}
\end{equation*}
The difference is $\abs{\frac{\mathrm{Re}\prntt{e^{i\theta}(a+b)}}{2}}$, but we can improve it to $\frac{\abs{a+b}}{2}$ by choosing $\theta=-\mathrm{arg}(a+b)$. \newline
\mbox{}\hfill $\square$

For part $(ii)$, we  use the same circuit as in Fig. \ref{fig:perfect2}. \ Let
\begin{equation*}
\begin{split}
    A\ket \psi &= \sqrt{p} \ket{1}\ket{\psi_1}+\sqrt{1-p}\ket{0}\ket{\psi_0}
    \\
    A\ket \phi &= \sqrt{1-p +\Delta} \ket{0}\ket{\phi_0}+\sqrt{p-\Delta}\ket{1}\ket{\phi_1}.
\end{split}
\end{equation*}
Then,
\begin{equation*}
    A\ket x = \frac{1}{\sqrt{2}}\prntt{\sqrt{p} \ket{1}\ket{\psi_1}+\sqrt{1-p}\ket{0}\ket{\psi_0}+\sqrt{1-p+\Delta} \ket{0}\ket{\phi_0}+\sqrt{p-\Delta}\ket{1}\ket{\phi_1}},
\end{equation*}
which after a phase flip on the first qubit ($Z_1$) yields
\begin{equation*}
    Z_1A\ket x =\frac{1}{\sqrt{2}}\prntt{-\sqrt{p} \ket{1}\ket{\psi_1}+\sqrt{1-p}\ket{0}\ket{\psi_0}+\sqrt{1-p+\Delta} \ket{0}\ket{\phi_0}-\sqrt{p-\Delta}\ket{1}\ket{\phi_1}}.
\end{equation*}
Meanwhile, 
\begin{equation*}
        A\ket y = \frac{1}{\sqrt{2}}\prntt{\sqrt{p} \ket{1}\ket{\psi_1}+\sqrt{1-p}\ket{0}\ket{\psi_0}-\sqrt{1-p+\Delta} \ket{0}\ket{\phi_0}-\sqrt{p-\Delta}\ket{1}\ket{\phi_1}}.
\end{equation*}
The inner product is the following:
\begin{equation*}
\begin{split}
    \bra y A^\dagger Z_1 A \ket x =
    \bra x A^\dagger Z_1 A \ket y^* =
    \frac{1}{2}&\Big[-p+(1-p) -(1-p+\Delta)+(p-\Delta)
    \\
    &+\sqrt{p(p-\Delta)}(\braket{\phi_1}{\psi_1}-\braket{\psi_1}{\phi_1})
    \\
    &+\sqrt{(1-p)(1-p+\Delta)}(\braket{\psi_0}{\phi_0}-\braket{\phi_0}{\psi_0})\Big].
\end{split}
\end{equation*}
Hence,
\begin{equation*}
    \frac{\abs{\bra y A^\dagger Z_1 A \ket x + \bra x A^\dagger Z_1 A \ket y}}{2} = \Delta.
\end{equation*}

\end{proof}
One implication of Theorem \ref{thm:main}$(i)$ is that if $U\ket x=\ket y$, while $U\ket y$ is orthogonal to both $\ket x$ and $\ket y$, then we can distinguish $\ket \psi$ from $\ket \phi$ with bias $\frac{1}{2}$ (because $\abs{a+b} =1$).

Another implication is that, if $U\ket{x}=\ket y$ but $U\ket y = -\ket x$, then $\abs{a+b}=0$ and we get no distinguishing power at all by the method of Theorem \ref{thm:main}$(i)$. \ One might wonder: is this just an artifact of our proof, or are there actual examples of $\ket x$ and $\ket y$ that are easy to swap with a $-1$ phase, but exponentially hard to swap with any other phase (or equivalently, for which it's exponentially hard to distinguish $\frac{\ket x + \ket y}{\sqrt{2}}$ from $\frac{\ket x - \ket y}{\sqrt{2}}$)? \ Perhaps surprisingly, we will show in Section \ref{sec:optimality} that the answer is the latter.

We stress that the ability to swap $\ket x$ and $\ket y$ is \emph{not} equivalent to the ability to distinguish $\ket{x}$ and $\ket{y}$ themselves. \ For example, let $\ket{x}=\ket{0\dots 0}$ and $\ket{y}=\sum_{j\neq 0}\alpha_j\ket{j}$, wherein $\prnttt{\alpha_j}$ are arbitrary (e.g., $\ket y$ is a Haar-random state). \ Then it's trivial to distinguish $\ket{x}$ from $\ket{y}$, yet \emph{mapping} $\ket{x}$ to $\ket{y}$ will in general be extremely difficult. \ Conversely, let $\ket{\psi}=\frac{\ket x+\ket y}{\sqrt 2}$ and $\ket{\phi}=\frac{\ket{x}-\ket{y}}{\sqrt 2}$. \ Then it's trivial to \emph{map} $\ket \psi$ to $\ket{\phi}$ and vice versa. \ But we know, by Theorem \ref{thm:main}, that \emph{distinguishing} the two must in general be extremely difficult. 

The general rule is that distinguishability in one basis implies ``swappability'' in a conjugate basis, and vice versa. \ (Note that Theorem \ref{thm:main} would also have worked with, e.g., $\ket{\psi}=\frac{\ket x +i\ket y}{\sqrt 2}$ and $\ket{\phi}=\frac{\ket x -i\ket y}{\sqrt 2}$, rather than $\ket{\psi}=\frac{\ket x +\ket y}{\sqrt 2}$ and $\ket{\phi}=\frac{\ket x -\ket y}{\sqrt 2}$.)

\section{Tightness\label{sec:optimality}}

A natural question about Theorem \ref{thm:main} is whether a better construction could yield better parameters. \ For example, given a $U$ such that $U\ket{x}=\ket y$ and $U\ket y$ is orthogonal to both $\ket x$ and $\ket y$, could we use it to distinguish $\ket{\psi}=\frac{\ket x + \ket y}{\sqrt 2}$ from $\ket{\phi}=\frac{\ket x - \ket y}{\sqrt 2}$ \emph{perfectly}, and with the same complexity?

We now prove that, in general, the parameters of Theorem \ref{thm:main} are optimal. \ Perhaps surprisingly, this optimality result does not depend on any unproved conjectures in complexity theory.

\begin{theorem} ~\\
\label{thm:optimality} 
\vspace*{-\baselineskip} % patch
\begin{enumerate} [label=(\roman*)]
    \item For all $0\le b\le a \le 1$, there exists an $n$-qubit $U$ implemented by a size-$O(1)$ circuit, as well as states $\ket{x}, \ket{y}$, such that $\bra y  U \ket x=a$ and $\bra x  U \ket y=b$, and yet if $\bra y  V \ket x=a'$ and $\bra x  V \ket y=b'$ where $|a'+b'|\ge|a+b|+\omega(2^{-n/3}\sqrt{\log n})$, then $V$ requires a size-$\omega(2^{n/3})$ circuit. 
    \item For all $\Delta \in [0,1]$, there exist two $n$-qubit states $\ket{\psi},\ket{\phi}$ that can be distinguished with bias $\Delta$ by a size-$O(1)$ circuit, yet such that distinguishing them with bias $\Delta+\omega(2^{-n/3}\sqrt{\log n})$ requires a size-$\omega(2^{n/3})$ circuit.
\end{enumerate}
We assume here that the size of the universal set of gates is polynomial in the number of qubits.
\end{theorem}
\begin{proof}
By our main result, Theorem \ref{thm:main}, we only need to prove part (i). \ Part (ii) then follows automatically, if we set $\ket\psi=\frac{\ket x + \ket y}{\sqrt{2}}$ and $\ket\phi=\frac{\ket x - \ket y}{\sqrt{2}}$ and $\Delta = \frac{\abs{a+b}}{2}$.

Let $\ket{\eta_0},\dots,\ket{\eta_7}$ be $n$-qubit states, whose pairwise swap complexity is exponential. \ For example, let them be Haar-random; then by a counting argument, all the pairwise swap complexities will clearly be exponential with overwhelming probability.

We add a 3-qubit index register to each $\ket{\eta_k}$, and write the entire state as $\ket{\bar k}\triangleq \ket{k}\otimes \ket{\eta_k}$. \ Consider the following construction:
\begin{equation} \label{eq:TightnessConstruction}
\begin{split} 
\ket{x}&=\sqrt{a-b} \prnt{  \frac{\ket{\bar 0}+\ket{\bar 1} +\ket{\bar 2}+\ket{\bar 3}}{2}} + \sqrt{b}\prnt{\frac{\ket{\bar 4}+\ket{\bar 5}}{\sqrt 2}} +\sqrt c \ket{\bar 6}
\\
\ket{y}&=\sqrt{a-b} \prnt{\frac{\ket{\bar 0}+i\ket{\bar 1} -\ket{\bar 2}-i\ket{\bar 3}}{2} }+ \sqrt{b}\prnt{\frac{\ket{\bar 4}-\ket{\bar 5}}{\sqrt 2}}+\sqrt c \ket{\bar 7}.
\\
U&= \prnt{\sum_{k=0}^3 i^k \ketbra{\bar k}{\bar k}+\ketbra{ \bar4}{ \bar4} -\ketbra{ \bar5}{ \bar5} +\ketbra{ \bar6}{ \bar6}+\ketbra{ \bar7}{ \bar7}}\otimes \mathbbm{1}_{2^n},
\end{split}
\end{equation}
wherein by normalization, $\braket{x}{x}=a+c=1$. \ To understand the construction, note that $U$ transfers the superposition of the first four states of $\ket x$ to the corresponding superposition in $\ket y$. \ In contrast, it transfers the superposition of the first four states in $\ket y$ to a superposition orthogonal to both $\ket x$ and $\ket y$. \ Next, $U$ applies $\ket{\bar 4}+\ket{\bar 5}\longleftrightarrow \ket{\bar 4}-\ket{\bar 5}$, and finally, $U$ does not affect the states $\ket{\bar 6}$ and $\ket{\bar 7}$.    

$U$ can be implemented using $O(1)$ gates since it acts only on the index register. \ Furthermore, it is easy to verify that 
$\bra{y}U\ket{x}=a$ and $\bra{x}U\ket{y}=b$. 

We will need the following lemma, proved in Appendix \ref{sec:ProofLemHaar}. 
\begin{lemma}\label{lem:Haar projection}
Let $\ket{\eta_0},\ket{\eta_1}$ be two $n$-qubit Haar-random states, and let $g=n^{O(1)}$ be the size of a universal set of gates $G$. Then with $1-\exp(-\exp(n))$ probability over $\ket{\eta_0},\ket{\eta_1}$, there is no circuit $C$ with  $M=O(2^{n/3})$ gates from $G$, such that $\abs{\bra {\eta_0} C \ket{\eta_1}}\ge\varepsilon$, where $\varepsilon \le  \sqrt{M\log g/N}=O(2^{-n/3}\sqrt{\log n}).$
%Let $\ket{\eta_0},\ket{\eta_1}$ be two Haar-random $n$-qubit states. \ Then with overwhelming probability, there is no circuit $C$ of size $O(2^{n/5})$ such that $\abs{\bra {\eta_0} C \ket{\eta_1}}=\Omega(2^{-n/3})$. 
\end{lemma}

By Lemma \ref{lem:Haar projection}, due to the pairwise swap-complexity of the $\prnttt{\ket{\eta_k}}$ states, any   unitary $\widetilde U$, implemented by $O(2^{n/3})$ gates, with overwhelming probability cannot transform $\ket{\bar k}$ into anything close to $\ket{\bar k'}$, for any $k\neq k'$. \ Hence, the principal submatrix of $\widetilde U$ when removing all columns except those corresponding to the $\prnttt{\ket{\bar k}}$ states is necessarily in the following form:
$$
\widetilde U \bigg\rvert_{\prnttt{\ket{\bar k}}} = \sum_{k=0}^7 \beta_k e^{i\theta_k} \ketbra{\bar k}{\bar k} + \widetilde O(2^{-n/3}),
$$
wherein $\beta_k\in[0,1]$, and $\widetilde O$ is a big-$O$ notation which ignores logarithmic factors. \ Calculating $\widetilde a+ \widetilde b$,
\begin{equation*}
    \begin{split}
\widetilde a &= \bra{y}\widetilde U \ket{x}= \frac{a-b}{4}\sum_{k=0}^3 \beta_k(-i)^k e^{i\theta_k} + \frac{b}{2}(\beta_4 e^{i\theta_4}-\beta_5 e^{i\theta_5})+\widetilde O(2^{-n/3}),
\\
\widetilde b &= \bra{x}\widetilde U \ket{y}= \frac{a-b}{4}\sum_{k=0}^3 \beta_k i^k e^{i\theta_k} + \frac{b}{2}(\beta_4 e^{i\theta_4}-\beta_5 e^{i\theta_5}) + \widetilde O(2^{-n/3}).
    \end{split}
\end{equation*}
By carefully choosing $\theta_k$, and taking $\beta_k=1$, we  upper-bound the absolute sum:
$$
\abs{\widetilde a+\widetilde b}=\abs{\frac{a-b}{2}\prnt{\beta_0 e^{\theta_0}-\beta_2 e^{i\theta_2}}+ b\prnt{\beta_4 e^{i\theta_4}-\beta_5 e^{i\theta_5}}}+\widetilde O(2^{-n/3}) \le \abs{a+b}
+\widetilde O(2^{-n/3}),$$
which proves the theorem. 

\end{proof}

Theorem \ref{thm:optimality} implies, in particular, that there exist cases where we can efficiently map $\ket x$ to $\ket y$, but only via a circuit that also maps $\ket y$ to $-\ket x$, and where eliminating the $-1$ factor requires an exponentially larger circuit. \ In these cases, and \textit{only} in these cases, we get efficient mapping between $\ket x$ and $\ket y$, without any corresponding ability to distinguish $\frac{\ket x + \ket y}{\sqrt{2}}$ from $\frac{\ket x - \ket y}{\sqrt{2}}$ efficiently.

Having said that, in the \textit{specific} case of Schr\"{o}dinger's cat, if we have the ability to map $\ket{\mathrm {Alive}}$ to $\ket{\mathrm {Dead}}$ and $\ket{\mathrm {Dead}}$ to $-\ket{\mathrm {Alive}}$, then we also have the ability to swap the $\ket{\mathrm {Alive}}$ and $\ket{\mathrm {Dead}}$ states without the $-1$ relative phase. \ The reason is that it's easy enough to \textit{distinguish} a live cat from a dead one, so we could simply correct the phase after applying the swap, conditional on being in the $\ket{\mathrm {Alive}}$ state. \ We thus see that, in the proof Theorem \ref{thm:optimality}, it was crucial to consider pairs of states that are exponentially hard not only to swap, but \textit{also} to distinguish from each other. \ (We thank Ed Witten for this observation.)

\section{Relative state complexity vs.\ swap complexity \label{sec:RelCvsSwapC}}

The following corollary summarizes the relation between the circuit complexity of swapping two states, their relative complexity and their absolute state complexity. 
\begin{corollary} \label{col:inequalities}
Consider two orthogonal states $\ket{x},\ket y$, and let $\ket{\psi}=\frac{\ket x+\ket y}{2}$ and $\ket{\phi}=\frac{\ket x - \ket y}{2}$. \ Then,
\begin{equation*}
    \mathcal C (\ket{x},\ket y)\le 
    \mathcal S (\ket x, \ket y)
    \le
    \min\big[{\mathcal C (\ket\psi,\ket{0\dots0} ),\mathcal C (\ket{\phi},~\ket{0\dots 0})}\big]  %\le \mathcal C(\ket x)+\mathcal C(\ket y)
\end{equation*}
(ignoring constant factors and $\varepsilon$ the subscripts of $\mathcal C,
\mathcal S$). The separation of the inequalities can be exponential.  
\end{corollary}
\begin{proof}
The first inequality is trivial, since swapping is at least as hard as mapping. \ For the second inequality, recall by Theorem \ref{thm:main} that $\mathcal S (\ket x, \ket y)$ is at most the complexity of distinguishing $\ket \psi$ from $\ket \phi$. Let $A$ be a minimal circuit to prepare $\ket\psi$ from $\ket{0^n}$. \ Then to distinguish $\ket \psi$ from $\ket \phi$, we simply apply $A^\dagger$ to $\ket \psi$, and check whether we got back to $\ket{0^n}$ (and similarly given a minimal circuit to prepare $\ket\phi$).

For an exponential separation of the first inequality, consider Equation \ref{eq:TightnessConstruction} with $b=c=0$. \ The unitary $U$ transfers $\ket x$ to $\ket y$ efficiently, hence  $\mathcal C(\ket{x,y})= O(1)$. \ On the other hand, $\mathcal S(\ket{x,y})$ can be exponential due to the tightness theorem (Theorem \ref{thm:optimality}). \ For an exponential separation of the second inequality, consider $\ket{x}=\ket{0}\ket{\eta}$ and $\ket{y}=\ket{1}\ket{\eta}$ where $\ket{\eta}$ is a Haar-random state. \ Then $\mathcal S(\ket{x},\ket{y})=O(1)$, but preparing either $\ket x$ or $\ket y$ requires an exponentially large complexity with overwhelming probability.
\end{proof}

Interestingly, while relative state complexity is a metric, swap complexity is not. \ Swap complexity is a ``semimetric'': it's symmetric and reflexive, but does not satisfy the triangle inequality, as shown by the following counterexample.

Consider the following $3$ states:
\begin{equation}
    \ket{x}=\ket{000}\quad\quad \ket{y}=\ket{1--}\quad\quad \ket{z}=\ket{011},
\end{equation}
and the universal set of gates: Hadamard, NOT, CNOT and a phase gate $R_{\phi}=\ketbra{0}{0}+e^{i\phi}\ketbra{1}{1}$, with $\phi\ll1$. \ It is easy to see that 
\begin{equation}
\mathcal S_0(\ket x,\ket z)=2 \quad \quad \mathcal S_0(\ket z,\ket y)=3,
\end{equation}
where $\mathcal S_0$ is the complexity of swapping  $\ket x$ and $\ket y$ with zero error. \
Our exhaustive search found that the smallest circuit for  perfectly swapping $x$ and $y$ is of size $7$ (see Figure \ref{fig:counterexample}).
\begin{figure} [H] 
{
\[
\Qcircuit @C=0.7em @R=1em {
\lstick{}  & \gate{X} & \qw & \qw & \qw & \qw & \qw
\\
\lstick{}  & \gate{X} & \qw&\gate{H}&\qw&\gate{X} & \qw
\\
\lstick{}  & \gate{X} & \qw&\gate{H}&\qw&\gate{X} & \qw
}  
\]
}
\caption{\label{fig:counterexample}A 7-gate circuit to swap $    \ket{000}$ and $\ket{1--}$. }
\end{figure}
Hence, 
\begin{equation}
\mathcal S_0(\ket x,\ket y)>\mathcal S_0(\ket x,\ket z)+\mathcal S_0(\ket z,\ket y)
\end{equation}

Note that swap complexity does satisfy the triangle inequality in the special case where $\ket{x},\ket{y},\ket{z}$ are all computational basis states.

\section{Discussion}
By using quantum circuit complexity, we were able to formalize a folklore observation in the foundations of quantum mechanics: namely, that the ability to measure the coherence in Schr\"{o}dinger's cat is somehow related to the ability to bring a dead cat back to life. \ We were also able to articulate in precisely what circumstances that folklore observation would become false. \ Our results inspired a more general investigation of \textit{swap complexity} of pairs of quantum states, which is related to their relative complexity but can be exponentially greater, and which might be independent interest.

Our equivalence theorem has some interesting implications for physics. \ For example, if we have a superposition of a state $\ket{x}$ of polynomial complexity and a state $\ket{y}$ of exponential complexity, then no polynomial-time experiment can ever detect the relative phase between $\ket{x}$ and $\ket y$. \ (For otherwise, we could efficiently \textit{map} $\ket{x}$ to $\ket y$!)

In a previous work \cite{Aaronson16}, Aaronson and Susskind proved that evolving the state 
$$\ket{\psi_0}=\frac{1}{\sqrt{2^n}} \sum_{j\in\prnttt{0,1}^n} \ket{j}\otimes\ket j$$
by a ``generic'' (computationally universal) Hamiltonian $H$ for exponential time yields a state with superpolynomial circuit complexity unless $ \mathsf{\PSPACE}\subset \mathsf{\PP/\poly}$. \ Combining that result with our Theorem \ref{thm:main} and Corollary \ref{col:inequalities} means that unless $ \mathsf{\PSPACE}\subset \mathsf{\PP/\poly}$, there can be no feasible experiment, in general, to measure the phase between a state and same state after being evolved for exponential time. \ Even if mapping one state to the other is merely ``thermodynamics-hard,'' in the sense that it's hard to unscramble an egg, still, distinguishing the superposition from the incoherent mixture with any non-negligible bias would be thermodynamics-hard as well.% \ Arguably, that might help to justify the state-dependent observable paradigm \cite{PR14a,PR14b,PR15}.

One might wonder about the apparent symmetry of our results in the case of Sch\"odinger's cat, since reviving a cat seems so much harder than taking its life. \ However note that in this work, both $\ket{\mathrm {Alive}}$ and $\ket{\mathrm {Dead}}$ are taken to determine the exact states of every atom of the cat. \ If we accounted for other possible ``alive'' and ``dead'' states, then of course we expect many more configurations of dead cats than alive cats, so thermodynamics suffices to explain why killing a cat is so much easier than reviving one.

\section{Acknowledgments}
We thank Edward Witten for the comment at the end of Section \ref{sec:optimality}, Daniel Gottesman for the note in the proof of Theorem \ref{thm:perfect}, Yosi Avron for pointing us to  the question of symmetry of swap complexity and relative state complexity in the discussion, and Henry Yuen  for catching some errors in an earlier draft.

\bibliography{bib}

\appendix

\section{Appendix: Proof of Lemma \ref{lem:Haar projection}\label{sec:ProofLemHaar}}
\setcounter{lemma}{0}
\begin{lemma}
Let $\ket{\eta_0},\ket{\eta_1}$ be two $n$-qubit Haar-random states, and let $g=n^{O(1)}$ be the size of a universal set of gates $G$. Then with $1-\exp(-\exp(n))$ probability over $\ket{\eta_0},\ket{\eta_1}$, there is no circuit $C$ with  $M=O(2^{n/3})$ gates from $G$ such that $\abs{\bra {\eta_0} C \ket{\eta_1}}\ge\varepsilon$, where $\varepsilon \le  \sqrt{M\log g/N}=O(2^{-n/3}\sqrt{\log n}).$
\end{lemma}

\begin{proof}
We use a simple counting argument. \ Starting at $\ket{\eta_0}$, a circuit with $M$ gates taken from a universal set of gates of size $g=n^{O(1)}$ reaches at most $O(g^M)$ different states $\prnttt{\ket {\gamma_j}}$. \
The following fact yields the probability of $\abs{\braket{\eta_0}{\gamma_j}}\ge\varepsilon$ for a specific $j$.
\begin{fact} [see Lemma 3.6 in \cite{AK07}]
Let $\ket \psi$ be a Haar-random state of dimension $N$. \ Then for any $\varepsilon >0$,
\begin{equation}
    \Pr \prnt{\abs{\braket{ \psi}{{0\dots 0}}}\ge \varepsilon}=  (1-\varepsilon^2)^{N-1}.
\end{equation}
\end{fact}
By the union bound, the probability that $\ket{\eta_1}$ has an overlap at least $\varepsilon$ with any of the states $\prnttt{\ket {\gamma_j}}$ is at most
\begin{equation} \label{eq:temp}
 {{g^M (1-\varepsilon^2)^{N-1}}} \le  
{{g^M e^{-\varepsilon^2(N-1)}}} \le  
{{2^{ M\log (g)-\log (e)\cdot M\log (g) (1-1/N)}}}\le 2^{M\log(g) \frac{1-\log(e)}{2}},
\end{equation}
where the first inequality comes from $(1+x)\le e^x$. \ The exponent in the last expression is of order $-2^{n/3}\log n$. \ Hence, the probability of any $\ket{\gamma_j}$ to  have at least $\varepsilon$ overlap with $\ket{\eta_1}$ is doubly-exponentially small.
\end{proof}

\end{document}